\documentclass[12pt]{amsart}
\usepackage[utf8]{inputenc}
\usepackage{amssymb, amsmath}
\usepackage{fourier}
\usepackage{xcolor}
\usepackage{enumitem}
\usepackage{fancyvrb}

\usepackage{geometry}
\theoremstyle{plain}

\newtheorem{theorem}{Theorem}

\theoremstyle{remark}

\newcommand{\bit}{\begin{itemize}}
\newcommand{\eit}{\end{itemize}}
\newcommand{\ben}{\begin{enumerate}}
\newcommand{\een}{\end{enumerate}}
\newcommand{\be}{\begin{equation}}
\newcommand{\ee}{\end{equation}}
\newcommand{\ba}{\begin{array}}
\newcommand{\ea}{\end{array}}

\newcommand{\bigO}[1]{\mathcal{O}\left( #1 \right)}

\newcommand{\dd}{\mathrm{d}}
\newcommand{\dt}{\mathrm{d}t}

\newcommand{\norm}[1]{\left|\left|#1\right|\right|}

\newcommand{\inner}[2]{\left< #1,#2 \right>}

\newcommand\cS{\mathcal S}
\newcommand\C{\mathbb C}
\newcommand\R{\mathbb R}

\newcommand{\fel}{\frac{1}{2}}

\newcommand\tr{\operatorname{Tr}}
\newcommand{\ler}[1]{\left( #1 \right)}

\title{The metric property of the quantum Jensen-Shannon divergence}

\author{D\'aniel Virosztek}
\address[D\'aniel Virosztek]{Institute of Science and Technology Austria\\
Am Campus 1, 3400 Klosterneuburg, Austria}
\email{daniel.virosztek@ist.ac.at}
\urladdr{http://pub.ist.ac.at/\~{}dviroszt}

\thanks{D. Virosztek was supported by the European Union's Horizon 2020 research and innovation programme under the Marie Sk\l odowska-Curie Grant Agreement No. 846294,  and partially supported by the Hungarian National Research, Development and Innovation Office (NKFIH) via grants no. K124152, and no. KH129601.
}

\dedicatory{To R\'oza Virosztek}


\begin{document}
\maketitle
\begin{abstract}
In this short note, we prove that the square root of the quantum Jensen-Shannon divergence is a true metric on the cone of positive matrices, and hence in particular on the quantum state space. 
\end{abstract}

\section{Introduction} \label{sec:intro}
\subsection{Motivation, goals} \label{susec:motiv-goals}
Quantum Jensen-Shannon divergence (QJSD) is among the most popular dissimilarity measures in information theory as it has a wide range of applications in the theory of complex networks \cite{dpa-jcn-3-2015, dnal-nature-2015}, pattern recognition \cite{brth-pattrec-48-2015}, graph theory \cite{rthw-pre-88-2013}, and chemical physics \cite{aal-jcp-130-2009}, beyond its applications in quantum theory \cite{dlh-pra-84-2011, rpjb-prl-116-2016, rfz-prl-105-2010}.
It is a symmetrized version of the quantum relative entropy, and hence several desirable properties of a divergence measure on quantum states (e.g., monotonicity under quantum maps) can be justified for the QJSD directly by the fact that they hold for the relative entropy. Although QJSD is a symmetric divergence measure, it is not a metric as it does not satisfy the triangle inequality. However, its square root is known to be a metric for commuting states and for pure states \cite{lmbcp-pra-77-2008}.
\par
Based on numerical investigations, it was conjectured more than ten years ago that the square root of the QJSD is a metric on the quantum state space in any finite dimension \cite{bh-pra-79-2009,lmbcp-pra-77-2008}. However, this conjecture has not been proved yet, as mentioned, for example, in \cite{dlh-pra-84-2011, dnal-nature-2015, mn-ijtp-53-2014, rpjb-prl-116-2016, rfz-prl-105-2010}.
Now we prove more: we prove that the square root of the quantum Jensen-Shannon divergence is a true metric on the cone of positive matrices in any dimension.
\par
Furtheremore, we will interpret a result of Carlen, Lieb, and Seiringer which tells us that on the qubit state space, the metric induced by the QJSD is a Hilbert space metric, and this is the case for a wide family of quantum Jensen divergences indexed by operator convex functions, as well.

\subsection{Notation, definition and basic properties of the QJSD} \label{susec:b-not-prop}
Throughout this note, $M_n^+(\C)$ stands for the cone of all positive semidefinite complex $n \times n$ matrices, $I$ denotes the identity matrix, and $\eta(x)=x \log{x}$ is the standard entropy function on $[0,\infty).$ Density matrices are positive semidefinite matrices with unit trace. 
\par
The most attractive quantum version of the classical Kullback-Leibler divergence is the Umegaki relative entropy, which is denoted by $S(.,.)$ and is defined for densities $\rho$ and $\sigma$ by
$$
S\ler{\rho, \sigma}=\tr \rho \ler{\log \rho -\log \sigma}.
$$
The Umegaki relative entropy is jointly convex \cite{lieb-aim-11-1973,lind-cmp-39-1974} and hence monotone under quantum maps, i.e., completely positive and trace preserving transformations of the state space \cite{lind-cmp-39-1974, uhlmann}.
\par
The quantum Jensen-Shannon divergence (QJSD) is denoted by $J(.,.)$ and is defined by
$$
J(\rho, \sigma)=\fel S \ler{\rho, \fel\ler{\rho+\sigma}}+\fel S \ler{\sigma, \fel\ler{\rho+\sigma}}.
$$
Note that the QJSD can be written in the form
$$
J(\rho,\sigma)=\fel\tr \eta(\rho)+\fel \tr \eta(\sigma)-\tr \eta \ler{\frac{\rho+\sigma}{2}} \qquad \ler{\rho,\sigma \in M_n^+(\C)},
$$
and we will prefer this latter form in the sequel. The QJSD is jointly convex, monotone under quantum maps, and the symmetry transformations of the state space and the positive cone with respect to the QJSD has been described in \cite{mn-ijtp-53-2014, mpv-laa-495-2016, v-lmp-106-2016}. It was conjectured more than ten years ago that the square root of the QJSD is a metric on the quantum state space \cite{bh-pra-79-2009,lmbcp-pra-77-2008}, but this has not been proved yet. In the next section, we give a proof of this metric property.
\section{Main result} \label{sec:main-res}

\begin{theorem} \label{thm:main}
The square root of the quantum Jensen-Shannon divergence given by
$$
J(A,B)=\fel\tr \eta(A)+\fel \tr \eta(B)-\tr \eta \ler{\frac{A+B}{2}} \qquad \ler{A,B \in M_n^+(\C)}
$$
is a true metric.
\end{theorem}

\begin{proof}
Symmetry, positivity, and definiteness is clear (the latter by the strict convexity of $\eta$), we have to show the triangle inequality.

\paragraph*{ \bf Step 1.}
The quantum Jensen-Shannon divergence admits the integral representation
\be \label{eq:int-rep}
J(A,B)=\int_0^{\infty}d_S^2(A+tI, B+tI) \dt \qquad \ler{A,B \in M_n^+(\C)}
\ee
where $d_S$ is the square root of the $S$-divergence defined by
$$
d_S^2(A,B)=-\fel\tr \log A-\fel\tr \log B+\tr \log \ler{\frac{A+B}{2}}
$$
for positive definite matrices \cite{Sra}.
On one hand,
$$
\frac{\dd}{\dt}J\ler{A+tI, B+tI}
=\frac{\dd}{\dt}\ler{\fel\tr \eta(A+tI)+\fel\tr \eta(B+tI)-\tr \eta \ler{\frac{A+B}{2}+tI}}
$$
$$
=\fel\tr \ler{\log(A+tI)+I}+\fel\tr \ler{\log(B+tI)+I}-\tr \ler{\log \ler{\frac{A+B}{2}+tI}+I}
$$
$$
=-d_S^2\ler{A+tI,B+tI}.
$$
On the other hand,
$$
\lim_{t \to \infty} J\ler{A+tI, B+tI}=0
$$
for every $A, B \in M_n^+(\C),$ because the Jensen-Shannon divergence is homogeneous, $\eta(1+r)=r+ \bigO{r^2},$ and hence
$$
J\ler{A+tI, B+tI}=t J \ler{I+\frac{A}{t}, I+\frac{B}{t}}
=t \ler{\fel\tr \eta \ler{I+\frac{A}{t}}+\fel\tr \eta \ler{I+\frac{B}{t}}-\tr \eta \ler{I+\frac{A+B}{2t}}}
$$
$$
=t\ler{\fel \tr \frac{A}{t}+\bigO{t^{-2}} +\fel \tr \frac{B}{t}+\bigO{t^{-2}} -\tr \frac{A+B}{2t}+\bigO{t^{-2}}}
=\bigO{t^{-1}}.
$$
So
$$
J(A,B)=\lim_{t \to \infty} \ler{J(A,B)- J(A+tI, B+tI)}=-\lim_{t \to \infty} \int_0^t \frac{\dd}{\dd r}J(A+rI, B+rI) \dd r
$$
$$
=\int_0^{\infty}d_S^2(A+rI, B+rI) \dd r
$$
as claimed.

\paragraph*{ \bf Step 2.}
The main result of \cite{Sra} is that $d_S$ is a metric on the positive definite cone. Therefore, for every $A,B,C \in M_n^+(\C)$ we have
$$
J(A,C)=\int_0^\infty d_S^2(A+tI,C+tI) \dt 
$$
$$
\leq \int_0^\infty d_S(A+tI, C+tI)\ler{d_S(A+tI, B+tI)+d_S(B+tI, C+tI)} \dt
$$
$$
\leq \sqrt{ \int_0^\infty d_S^2(A+tI, C+tI) \dt}\ler{\sqrt{ \int_0^\infty d_S^2(A+tI, B+tI) \dt}+\sqrt{ \int_0^\infty d_S^2(B+tI, C+tI) \dt}}
$$
$$
=\sqrt{J(A,C)}\ler{\sqrt{J(A,B)}+\sqrt{J(B,C)}}
$$
where the first inequality relies on Sra's result, and the second one follows from the Cauchy-Schwartz and the Minkowski inequalities.
Dividing by $\sqrt{J(A,C)}$ completes the proof.
\end{proof}

\section{The $2 \times 2$ case}
In this section we discuss the case of $2 \times 2$ matrices, because in this special case, one can prove more than the metric property of the square root of the QJSD in the following sense: on one hand, it can be shown that this metric is a Hilbert space metric, on the other hand, this is true not only for the QJSD, but for a large family of quantum Jensen divergences indexed by operator convex functions.
\par
In the sequel, I interpret the argument of Eric Carlen, Elliott Lieb, and Robert Seiringer that I learned from Robert Seiringer in 2017. The symbol $\cS\ler{\C^2}$ stands for the space of $2 \times 2$ density matrices.
\begin{theorem}[Carlen-Lieb-Seiringer]
For an operator convex function $f:[0,\infty) \rightarrow \R,$ the square root of the symmetric quantum Jensen $f$-divergence defined by
$$
J_f\ler{\rho, \sigma}:=\fel \ler{\tr f\ler{\rho}+ \tr f\ler{\sigma}}-\tr f \ler{\frac{\rho+\sigma}{2}} \qquad \ler{\rho, \sigma \in \cS\ler{\C^2}}
$$
is a true metric on $\cS\ler{\C^2},$ moreover, it admits a Hilbert space embedding. 
\end{theorem}
\begin{proof}
Operator convex functions on $[0, \infty)$ admit the integral representation
$$
f(x)=a+bx+c x^2+\int_0^\infty \ler{\frac{x}{1+t}-\frac{x}{x+t}} \dd \mu (t)
=a+bx+c x^2+\int_0^\infty \ler{\frac{x}{1+t}-1+\frac{t}{x+t}} \dd \mu (t), 
$$
where $a,b \in \R, c\geq 0,$ and $\mu$ is a positive measure on $(0, \infty)$ with $\int_{(0,\infty)}(1+t)^{-2} \dd \mu (t)< +\infty.$ (see, e.g., \cite[8.1]{hmpb-rmp-2011}).
\par
Jensen divergences are linear in the sense that $J_{\alpha f +g}(\cdot,\cdot)=\alpha J_{f}(\cdot,\cdot)+J_{g}(\cdot,\cdot).$
Affine functions do not matter, $J_{a}\ler{\cdot,\cdot}\equiv 0$ for $a(x)=\alpha x+\beta.$ Note that the quadratic function $q(x)=x^2$ gives the Hilbert-Schmidt norm square, $J_q\ler{\rho,\sigma}=\tfrac{1}{4}\norm{\rho-\sigma}_{HS}^2.$
\par
By Schoenberg's theorem on the Hilbert space representation of negative definite kernels (see the original works \cite{sch-aom-39-1938,sch-tams-44-1938} or \cite[Chapter 3, 3.2.]{bcp-book}), it suffices to prove that $J_{f}\ler{\cdot,\cdot}$ is negative definite on $\cS\ler{\C^2},$ that is,
\be \label{eq:cnd-def}
\sum_{j,k=1}^m c_j c_k J_f\ler{\rho_j, \rho_k} \leq 0
\ee
for all $c_1,\dots,c_m \in \R$ with $\sum_{j=1}^m c_j=0$ and for all $\rho_1, \dots, \rho_m \in \cS\ler{\C^2}.$
This clearly holds for $J_q\ler{\cdot,\cdot},$ so we need to prove it only for $J_{f_t}\ler{\cdot,\cdot},$ where $f_t(x)=\frac{t}{t+x}$ and $t>0.$
\par
Let $\mathbf{B}$ denote the closed unit ball in $\R^3,$ and let $r_j \in \mathbf{B}$ denote the Bloch vector of $\rho_j,$ that is, $\rho_j=\fel \ler{I+r_j \cdot \sigma},$ where $\sigma=\left\{\sigma_x, \sigma_y, \sigma_z\right\}$ is the vector containing the Pauli matrices. Now

$$
\sum_{j,k=1}^m c_j c_k J_{f_t} \ler{\rho_j, \rho_k}
=-\sum_{j,k=1}^m c_j c_k \tr \frac{t}{t+\frac{\rho_j+\rho_k}{2}}
$$
$$
=
-\sum_{j,k=1}^m c_j c_k \ler{\frac{t}{t+\tfrac{1}{2}\ler{1+\norm{\tfrac{r_j+r_k}{2}}}}+\frac{t}{t+\tfrac{1}{2}\ler{1-\norm{\tfrac{r_j+r_k}{2}}}}}
$$
$$
=
-\sum_{j,k=1}^m c_j c_k \ler{\frac{2 t^2+t}{\ler{t+\tfrac{1}{2}}^2-\norm{\tfrac{r_j+r_k}{4}}^2}}
=
-\sum_{j,k=1}^m c_j c_k \ler{\frac{2 t^2+t}{\ler{t+\tfrac{1}{2}}^2}}\ler{\frac{1}{1-\norm{\tfrac{r_j+r_k}{4t+2}}^2}}
$$

where $\norm{.}$ is the Euclidean norm in $\R^3.$
Moreover,
$$
\frac{1}{1-\norm{\tfrac{r_j+r_k}{4t+2}}^2}=
\int_{0}^\infty \exp\ler{-s \ler{1-\norm{\tfrac{r_j+r_k}{4t+2}}^2}} \dd s
$$
$$
=
\int_{0}^\infty e^{-s}\exp\ler{ s \norm{\tfrac{r_j}{4t+2}}^2}\exp\ler{ s \norm{\tfrac{r_k}{4t+2}}^2}
\exp\ler{ s \frac{2 \inner{r_j}{r_k}}{\ler{4t+2}^2}}\dd s.
$$
Therefore,
\be \label{eq:final}
\sum_{j,k=1}^m c_j c_k J_{f_t} \ler{\rho_j, \rho_k}=-\frac{2 t^2+t}{\ler{t+\tfrac{1}{2}}^2} \int_0^{\infty} e^{-s} \sum_{j,k=1}^m c_j(s) c_k(s) \exp\ler{\frac{2 s \inner{r_j}{r_k}}{\ler{4t+2}^2}} \dd s,
\ee
where $c_j(s):=c_j \exp{\ler{s \norm{\tfrac{r_j}{4t+2}}^2}}.$

The right hand side of \eqref{eq:final} is non-positive for any $c_1, \dots, c_m \in \R$ and $r_1, \dots r_m \in \mathbf{B}$ as the product of positive definite functions is positive definite, and all the coefficients of the power series of the exponential function are positive.
The proof is done.

\end{proof}

We note that Bri\"et and Harremo\"es gave a proof of the Hilbert space embedding property for $2 \times 2$ densities and for a special class of operator convex functions containing the standard entropy function \cite{bh-pra-79-2009}. However, their proof contains a mistake, namely, in the proof of \cite[Lemma 4]{bh-pra-79-2009}, which is the key step in their argument, it is erroneously claimed that the function
$$
\cS\ler{\C^2} \times \cS\ler{\C^2} \rightarrow \R; \qquad \ler{\rho,\sigma} \mapsto 2 \tr \ler{\fel\ler{\rho+\sigma}}^2 -1
$$
is positive definite. For example, if
$$
\rho=\left[\ba{cc} 1 & 0 \\ 0 & 0 \ea\right] \, \text{ and } \, \sigma=\left[\ba{cc} \fel & 0 \\ 0 & \fel \ea\right],
$$
then the matrix
$$
\left[\ba{cc} 2 \tr \rho^2 -1 & 2 \tr \ler{\fel\ler{\rho+\sigma}}^2 -1 \\ 2 \tr \ler{\fel\ler{\rho+\sigma}}^2 -1 & 2 \tr \sigma^2 -1 \ea\right]
=\left[\ba{cc} 1 & \frac{1}{4} \\ \frac{1}{4} & 0 \ea\right]
$$
is indefinite.

\par
We also note that it is known that the cone of $2 \times 2$ positive definite matrices equipped with the Jensen divergence corresponding to the operator convex function $x \mapsto -\log x$ does not admit any Hilbert space embedding \cite{Sra}. This fact shows that the result of Carlen, Lieb, and Seiringer is optimal in the sense that one can not go beyond the qubit state space in proving the Hilbert space embedding property in its full generality, that is, for all operator convex functions. (To be precise, $x \mapsto - \log x$ is operator convex only on $(0,\infty)$ and not on $[0, \infty),$ but the integral representation $\log{x}=\int_0^{\infty} \ler{\frac{1}{1+t}-\frac{1}{x+t}} \dt$ shows that from the viewpoint of negative definite kernels, it behaves similarly to operator convex functions on $[0,\infty).$)
However, for the standard entropy function $\eta(x)=x \log{x},$ we do not have such counterexamples.
So it is still an open question, whether the metric induced by the QJSD is a Hilbert space metric on the cone of positive semidefinite $n \times n$ matrices for $n \geq 2,$ or on the space of $n \times n$ density matrices for $n \geq 3.$

\subsubsection*{Acknowledgements}
I would like to thank Robert Seiringer for great discussions on the topic and for showing me their argument on the case of $2 \times 2$ matrices. I am also grateful to Lajos Moln\'ar for proposing this problem to me around 2015, and to the anonymous referee for his/her valuable comments and suggestions.
\bibliographystyle{amsplain}

\end{document}